\documentclass[conference]{IEEEtran}
\usepackage{enumerate}
\usepackage{multirow,array}
\usepackage{cite}
\usepackage{graphicx}
\usepackage{psfrag}
\usepackage{caption}
\usepackage{subcaption}
\usepackage{url}
\usepackage{amsmath}
\usepackage{amsthm}
\usepackage{array}
\usepackage{amssymb}
\usepackage{amsfonts}
\usepackage{float}
\usepackage{tabu}
\usepackage{algorithm}
\usepackage{algorithmicx}
\usepackage{algcompatible}
\usepackage{algpseudocode}
\usepackage{color, colortbl}
\definecolor{Gray}{gray}{0.9}

\usepackage[font=small]{caption}
\usepackage{subcaption}
\usepackage{tikz}
\tikzset{
    cross/.pic = {
    \draw[rotate = 45] (-#1,0) -- (#1,0);
    \draw[rotate = 45] (0,-#1) -- (0, #1);
    }
}
\usetikzlibrary{patterns}
\usetikzlibrary{decorations.pathreplacing,calligraphy}

\allowdisplaybreaks

\newcommand{\floor}[1]{\left\lfloor #1 \right\rfloor}

\makeatletter
\renewcommand{\boxed}[1]{\text{\fboxsep=.2em\fbox{\m@th$\displaystyle#1$}}}
\makeatother
\newcommand{\bea}{\begin{eqnarray}}
\newcommand{\eea}{\end{eqnarray}}
\newcommand{\bean}{\begin{eqnarray*}}
\newcommand{\eean}{\end{eqnarray*}}

\newcommand{\B}{\boldsymbol}

\newcounter{casenum}

\newtheorem*{conjecture*}{Conjecture}
\newtheorem{lemma}{Lemma}

\newtheorem{definition}{Definition}
\newtheorem{theorem}{Theorem}

\newtheorem{constr}{Construction}
\IEEEoverridecommandlockouts
\begin{document}

\title{Rate-Optimal Streaming Codes Over the Three-Node Decode-And-Forward Relay Network}

\author{\IEEEauthorblockN{\textbf{Shubhransh~Singhvi}\IEEEauthorrefmark{1}, 
                     \textbf{Gayathri R.}\IEEEauthorrefmark{2} 
                     and \textbf{P. Vijay Kumar}\IEEEauthorrefmark{2}}
                     
    \IEEEauthorblockA{\IEEEauthorrefmark{1}%
                      {Signal Processing \&  Communications Research  Center, International Institute  of  Information Technology, Hyderabad}}
    
    \IEEEauthorblockA{\IEEEauthorrefmark{2}%
                     Electrical Communication Engineering, Indian Institute of Science, Bangalore 
                     }
                     
    email: shubhransh.singhvi@students.iiit.ac.in, \{gayathrir, pvk\}@iisc.ac.in
    
\thanks{This research is supported in part by the J C Bose National Fellowship JCB/2017/000017 and in part by SERB Grant No. CRG/2021/008479.}
 }
 \date{\today}
\maketitle
\thispagestyle{empty}	
\pagestyle{empty}

\begin{abstract}
We study the three-node Decode-and-Forward ($\mathsf{D\&F}$) relay network subject to random and burst packet erasures. The source wishes to transmit an infinite stream of packets to the destination via the relay. The three-node $\mathsf{D\&F}$ relay network is constrained by a decoding delay of $T$ packets, i.e., the packet transmitted by the source at time $i$ must be decoded by the destination by time $i+T$. For the individual channels from source to relay and relay to destination, we assume a delay-constrained sliding-window ($\mathsf{DCSW}$) based packet-erasure model that can be viewed as a tractable approximation to the commonly-accepted Gilbert-Elliot channel model. Under the model, any time-window of width $w$ contains either up to $a$ random erasures or else erasure burst of length at most $b~(\geq a)$. Thus the source-relay and relay-destination channels are modeled as $(a_1,b_1,w_1,T_1)$ and $(a_2,b_2,w_2,T_2)$ $\mathsf{DCSW}$ channels. We first derive an upper bound on the capacity of the three-node $\mathsf{D\&F}$ relay network. We then show that the upper bound is tight for the parameter regime: $\max\{b_1,b_2\}~|~(T-b_1-b_2-\max\{a_1,a_2\}+1), a_1=a_2 ~\text{OR} ~b_1=b_2$ by constructing streaming codes achieving the bound. The code construction requires field size linear in $T$, and has decoding complexity equivalent to that of decoding an $\mathsf{MDS}$ code. 
\end{abstract}
\section{Introduction}
Low-latency communication is a critical ingredient of upcoming promising applications such as telesurgery, virtual and augmented reality,  industrial automation and self-driving cars \cite{5GAmericas}. Ultra-Reliable, Low-Latency Communication (URLLC) is one of the three core focus areas of 5G, where latency is measured as the time elapsed between the transmission of a packet from the source and it's recovery at the receiver. Using ARQ-based schemes to ensure reliability results in an undesirable round-trip delay, which makes it challenging to meet the low latency requirement of URLLC. Physical layer FEC cannot help recover from packet drops arising due to congestion, a wireless link in deep fade or else late packet arrival. Streaming codes represent a packet-level FEC scheme for countering such packet losses. 


\subsection{A Brief History of Streaming Codes}
Research on streaming codes began with the study of burst-erasure correction under a decoding-delay constraint \cite{MartSunTIT04, MartTrotISIT07}; the authors argued in favor of packet-extension encoding framework, where the redundancy is added within the packets rather than being transmitted as separate packets to avoid adding to network congestion. A measurement study of mobile video calls over wireless networks \cite{YuVCStudyINFOCOM} indicated that packet erasures occur both in an isolated and bursty fashion. In \cite{BadrPatilKhistiTIT17}, a delay-constrained sliding-window ($\mathsf{DCSW}$) channel model was introduced as a tractable deterministic approximation to the commonly-accepted Gilbert-Elliott erasure channel model \cite{gilbert, elliott, HasHoh, VajRamJhaKum, RamVajJhaKum} that is capable of causing burst and random erasures. An $(a,b,w,T)$ $\mathsf{DCSW}$ channel imposes a decoding-delay constraint of $T$, and can cause at most $a$ random erasures or else, a burst of $b$ erasures within any sliding window of size $w$ time slots, where $0<a \leq b \leq T$. 
\begin{figure}[H]
\centering
\resizebox{.475 \textwidth}{!}
{
\begin{tikzpicture}

\fill (-0.2,0.5) circle[radius=1pt];
\fill (-0.4,0.5) circle[radius=1pt];
\fill (-0.6,0.5) circle[radius=1pt];

\draw[thick,<->] (0,1.2) -- (5,1.2);
\draw (2.5,1.4) node{$w = 5$; burst of length $3$};
\draw[thick,<->] (2,2) -- (7,2);
\draw (4.5,2.2) node{$w = 5$; $2$ random erasures};

\draw[step=1cm,black,very thin] (0,0) grid (10,1);

\path (.5,0.5) pic[red, rotate = 0] {cross=10pt};
\path (1.5,0.5) pic[red, rotate = 0] {cross=10pt};
\path (2.5,0.5) pic[red, rotate = 0] {cross=10pt};
\path (6.5,0.5) pic[red, rotate = 0] {cross=10pt};
\path (8.5,0.5) pic[red, rotate = 0] {cross=10pt};

\draw (0.5,-0.2) node{$t$};
\draw (1.5,-0.2) node{$t+1$};
\draw (2.5,-0.2) node{$t+2$};
\draw (3.5,-0.2) node{$t+3$};
\draw (4.5,-0.2) node{$t+4$};
\draw (5.5,-0.2) node{$t+5$};
\draw (6.5,-0.2) node{$t+6$};
\draw (7.5,-0.2) node{$t+7$};
\draw (8.5,-0.2) node{$t+8$};
\draw (9.5,-0.2) node{$t+9$};

\fill (10.2,0.5) circle[radius=1pt];
\fill (10.4,0.5) circle[radius=1pt];
\fill (10.6,0.5) circle[radius=1pt];

\draw (5,-0.75) node{time $\longrightarrow$};
\end{tikzpicture}
}
\caption{Illustrating a permissible erasure pattern in an $(a=2,b=3,w=5,T)~  \mathsf{DCSW}$ channel.} \label{fig:three-node}
\end{figure}
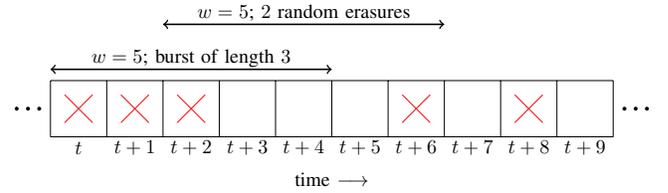

 \vspace{-2ex}
In \cite{BadrPatilKhistiTIT17, NikDeepPVK}, it was shown that one can without loss of generality, set $w=T+1$ and hence, we will abbreviate and write $(a,b,T)$ in place of $(a,b,w,T)$. A packet-level code is referred to as an $(a,b,T)$ streaming code if it enables recovery from all the permissible erasure patterns of an $(a,b,T)$ $\mathsf{DCSW}$ channel. The coding rate, denoted by $\mathsf{R}$, of an $(a,b,T)$ streaming code was shown in \cite{BadrPatilKhistiTIT17} to be upper bounded as 
\bea\label{eq:rbpp}
    \mathsf{R} \leq  \frac{T-a+1}{T-a+1+b} \triangleq \mathsf{C}_{a,b,T}.
\eea
Hence, we will refer to $\mathsf{C}_{a,b,T}$ as the point-to-point channel capacity. Code constructions achieving the point-to-point channel capacity for all parameters $(a,b,T)$ can be found in \cite{NikPVK, KhistiExplicitCode, Small}. The construction in \cite{NikPVK} is not explicit, and requires a field size of $q^2$, where $q \geq T+b-a$ is a prime power. \cite{KhistiExplicitCode} provides an explicit construction with a field size that scales quadratically with the delay. Explicit code construction with reduced quadratic field size, $q^2,$ where $q\geq T$ is presented in \cite{Small}. Streaming codes based on staggered diagonal embedding \cite{NikRamVajKum, RamVajNikPVK, GSDE} having linear field size are rate-optimal for special cases. Streaming codes have also been  constructed for channels  with  unequal  source-channel inter-arrival rates \cite{BadrPatilKhistiTIT17}, multiplicative-matrix channels \cite{robin} and multiplexed communication scenarios with different decoding delays for different streams \cite{KhistiMultiplex, KhistiMultiplex2}. In \cite{RudowRashmi18}, the authors consider a setting for variable-size arrivals. Locally recoverable streaming codes for packet-erasure recovery were constructed in \cite{RamVajKum}.  Other FEC schemes suitable for streaming setting can be found in \cite{AdlCas, LeongHo, LeoQurHo, MalMedYeh, bats, Shokrollahi, IyeSUW, FelZig, DraKhi}. In contrast to the existing literature on burst and random packet erasure correcting streaming codes which focuses on point-to-point networks, our focus in this paper is on three-node relay network, which consists of a source, a relay and a destination. This  kind  of topology is often present in content delivery networks \cite{KacTom, KhistiThreeNode, NikFacDomKhi}. 

 \vspace{-1ex}
\subsection{Paper Outline}
The three-node $\mathsf{D\&F}$ relay network and symbol-wise decode-and-forward $(\mathsf{SW~D\&F})$ strategy are introduced in Section \ref{sec:three-node}.
An upper bound on the capacity of the three-node $\mathsf{D\&F}$ relay network is derived in Section \ref{sec:upp-bnd}. The Staggered diagonal embedding $\left(\mathsf{SDE}\right)$ approach \cite{NikRamVajKum} is introduced in Section \ref{sec:SDE}. A rate-optimal streaming code construction using the $\mathsf{SDE}$ approach and the $\mathsf{SW~D\&F}$ strategy is provided in Section \ref{sec:constr}. Section \ref{sec:conclusion} concludes the paper. 
 
\section{Three-Node Relay Network}
\label{sec:three-node}
We follow \cite{KhistiThreeNode} in assuming a $\mathsf{D\&F}$ network and further, one in which the encoding function at the relay does not take into account the erasure pattern observed over $(s,r)$ channel~\footnote{This assumption is relaxed in the more recent work~\cite{FacKris}.}. Such a relaying strategy would be preferred in settings where the relay also has an interest in  the contents of the packet stream. The network consists of a source, a destination and a relay between them, which are denoted by $s, d$ and $r$, respectively. The channel between nodes $s$ and $r$ is denoted by $(s, r)$, and the channel between nodes $r$ and $d$ is denoted by $(r, d)$. We consider the case where the $(s,r)$ and $(r,d)$ channels are subject to both random and burst erasures. Thus the channels $(s,r)$ and $(r,d)$ are modeled as $(a_{1},b_{1},T_{1})$ and $(a_{2},b_{2},T_{2})$ $\mathsf{DCSW}$ channels, respectively, where $T_{1}, T_{2} < T$. 
 \vspace{-2ex}
 
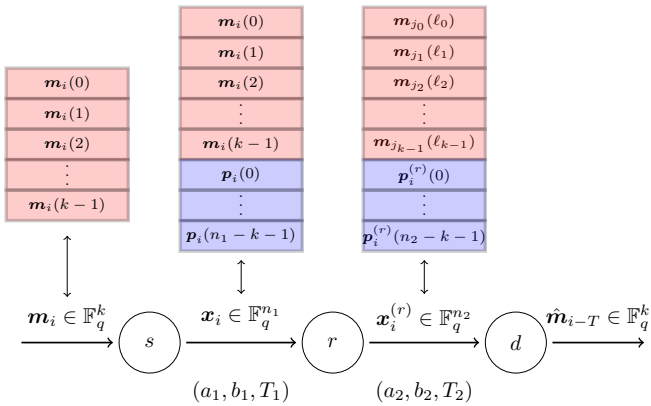
\begin{figure}[H]
\centering
\resizebox{.49 \textwidth}{!}
{
\begin{tikzpicture}

\draw [ultra thick, draw=black, fill=red, opacity=0.2] (-2.35,4) rectangle (-0.35,4.5); 
\draw [ultra thick, draw=black, fill=red, opacity=0.2] (-2.35,3.5) rectangle (-0.35,4); 
\draw [ultra thick, draw=black, fill=red, opacity=0.2] (-2.35,3) rectangle (-0.35,3.5); 
\draw [ultra thick, draw=black, fill=red, opacity=0.2] (-2.35,2.5) rectangle (-0.35,3); 
\draw [ultra thick, draw=black, fill=red, opacity=0.2] (-2.35,2) rectangle (-0.35,2.5); 
\draw (-1.35,4.25) node{\scriptsize{$\B{m}_{i}(0)$}};
\draw (-1.35,3.75) node{\scriptsize{$\B{m}_{i}(1)$}};
\draw (-1.35,3.25) node{\scriptsize{$\B{m}_{i}(2)$}};
\draw (-1.35,2.9) node{\scriptsize{$.$}};
\draw (-1.35,2.75) node{\scriptsize{$.$}};
\draw (-1.35,2.6) node{\scriptsize{$.$}};
\draw (-1.35,2.25) node{\scriptsize{$\B{m}_{i}(k-1)$}};

\draw [ultra thick, draw=black, fill=red, opacity=0.2] (0.5,5) rectangle (2.5,5.5); 
\draw [ultra thick, draw=black, fill=red, opacity=0.2] (0.5,4.5) rectangle (2.5,5); 
\draw [ultra thick, draw=black, fill=red, opacity=0.2] (0.5,4) rectangle (2.5,4.5); 
\draw [ultra thick, draw=black, fill=red, opacity=0.2] (0.5,3.5) rectangle (2.5,4); 
\draw [ultra thick, draw=black, fill=red, opacity=0.2] (0.5,3) rectangle (2.5,3.5); 
\draw [ultra thick, draw=black, fill=blue, opacity=0.2] (0.5,2.5) rectangle (2.5,3); 
\draw [ultra thick, draw=black, fill=blue, opacity=0.2] (0.5,2) rectangle (2.5,2.5); 
\draw [ultra thick, draw=black, fill=blue, opacity=0.2] (0.5,1.5) rectangle (2.5,2); 
\draw (1.5,5.25) node{\scriptsize{$\B{m}_{i}(0)$}};
\draw (1.5,4.75) node{\scriptsize{$\B{m}_{i}(1)$}};
\draw (1.5,4.25) node{\scriptsize{$\B{m}_{i}(2)$}};
\draw (1.5,3.9) node{\scriptsize{$.$}};
\draw (1.5,3.75) node{\scriptsize{$.$}};
\draw (1.5,3.6) node{\scriptsize{$.$}};
\draw (1.5,3.25) node{\scriptsize{$\B{m}_{i}(k-1)$}};
\draw (1.5,2.75) node{\scriptsize{$\B{p}_{i}(0)$}};
\draw (1.5,2.4) node{\scriptsize{$.$}};
\draw (1.5,2.25) node{\scriptsize{$.$}};
\draw (1.5,2.1) node{\scriptsize{$.$}};
\draw (1.5,1.75) node{\scriptsize{$\B{p}_{i}(n_1-k-1)$}};

\draw (4.5,5.25) node{\scriptsize{$\B{m}_{j_0}(\ell_0)$}};
\draw (4.5,4.75) node{\scriptsize{$\B{m}_{j_1}(\ell_1)$}};
\draw (4.5,4.25) node{\scriptsize{$\B{m}_{j_2}(\ell_2)$}};
\draw (4.5,3.9) node{\scriptsize{$.$}};
\draw (4.5,3.75) node{\scriptsize{$.$}};
\draw (4.5,3.6) node{\scriptsize{$.$}};
\draw (4.5,3.25) node{\scriptsize{$\B{m}_{j_{k-1}}(\ell_{k-1})$}};
\draw (4.5,2.75) node{\scriptsize{$\B{p}_{i}^{(r)}(0)$}};
\draw (4.5,2.4) node{\scriptsize{$.$}};
\draw (4.5,2.25) node{\scriptsize{$.$}};
\draw (4.5,2.1) node{\scriptsize{$.$}};
\draw (4.5,1.75) node{\scriptsize{$\B{p}_{i}^{(r)}(n_2-k-1)$}};
\draw [ultra thick, draw=black, fill=red, opacity=0.2] (3.5,5) rectangle (5.5,5.5); 
\draw [ultra thick, draw=black, fill=red, opacity=0.2] (3.5,4.5) rectangle (5.5,5); 
\draw [ultra thick, draw=black, fill=red, opacity=0.2] (3.5,4) rectangle (5.5,4.5); 
\draw [ultra thick, draw=black, fill=red, opacity=0.2] (3.5,3.5) rectangle (5.5,4); 
\draw [ultra thick, draw=black, fill=red, opacity=0.2] (3.5,3) rectangle (5.5,3.5); 
\draw [ultra thick, draw=black, fill=blue, opacity=0.2] (3.5,2.5) rectangle (5.5,3); 
\draw [ultra thick, draw=black, fill=blue, opacity=0.2] (3.5,2) rectangle (5.5,2.5); 
\draw [ultra thick, draw=black, fill=blue, opacity=0.2] (3.5,1.5) rectangle (5.5,2);

\draw[<->] (-1.35,0.75) -- (-1.35,1.75);
\draw[thick,->] (-2.1,0) -- (-0.6,0);
\draw (-1.35,0.4) node{$\B{m}_i \in \mathbb{F}_q^k$};
\draw (0,0) circle (0.5cm);
\draw (0,0) node{$s$};
\draw[thick,->] (0.6,0) -- (2.4,0);
\draw (1.5,0.4) node{$\B{x}_i \in \mathbb{F}_q^{n_{1}}$};
\draw[<->] (1.5,0.8) -- (1.5,1.35);
\draw (1.5,-0.8) node{$(a_1,b_1,T_1)$};
\draw (3,0) circle (0.5cm);
\draw (3,0) node{$r$};
\draw[thick,->] (3.6,0) -- (5.4,0);
\draw (4.5,0.4) node{$\B{x}_i^{(r)} \in \mathbb{F}_q^{n_{2}}$};
\draw[<->] (4.5,0.8) -- (4.5,1.35);
\draw (4.5,-0.8) node{$(a_2,b_2,T_2)$};
\draw (6,0) circle (0.5cm);
\draw (6,0) node{$d$};
\draw[thick,->] (6.6,0) -- (8.1,0);
\draw (7.35,0.4) node{$\hat{\B{m}}_{i-T} \in \mathbb{F}_q^k$};
\end{tikzpicture}
}
\caption{Illustrating a three-node $\mathsf{D\&F}$ relay network under packet-extension framework. Here, $(.)$ denotes the symbol index in a packet. The time instances $j_0,j_1,\ldots,j_{k-1}\in[0,i]$, and the symbol indices $\ell_0,\ell_1,\ldots,\ell_{k-1}\in[0,k-1]$. } \label{fig:three-node}
\end{figure}

\vspace{-1ex}
Node $s$ wishes to transmit an infinite stream of packets $\{\B{m}_i\}_{i=0}^{\infty}$ to node $d$ via the node $r$. At any time $i$, as shown in Fig. \eqref{fig:three-node}, the source encodes the packet $\B{m}_i \in \mathbb{F}_q^k$, into a coded packet $\B{x}_i = \left[\B{m}_i^{\intercal},\B{p}_i^{\intercal}\right]^{\intercal} \in \mathbb{F}_q^{n_{1}}$, where the parity $\B{p}_i\in\mathbb{F}_q^{n_1-k}$ in general depends upon all prior and current message packets $\{\B{m}_j\}_{j=0}^{i}$. The encoded packet is transmitted to the relay via the $(s,r)$ channel, and the relay receives $\B{y}_i^{(r)} \in \mathbb{F}_q^{n_{1}} \cup \{*\}$, where $\B{y}_i^{(r)}$ equals either $\B{x}_i$ or erasure symbol $``*"$. In the same time slot $i$, the relay transmits $\B{x}_i^{(r)} \in \mathbb{F}_q^{n_{2}}$ to the destination through the $(r,d)$ channel. The transmitted packet at the relay is permitted to be a function of all the received packets till time $i$; ${\{\B{y}^{(r)}_j\}_{j=0}^{i}}$ \footnote{We follow \cite{KhistiThreeNode} in adopting this convention.}. The destination receives $\B{y}_i \in \mathbb{F}_q^{n_{2}} \cup \{*\}$ where $\B{y}_i$ equals either $\B{x}_i^{(r)}$ or erasure symbol $``*"$. The decoding-delay constraint of $T$ is construed as requiring that destination must produce an estimate of $\B{m}_i$, denoted by $\hat{\B{m}}_{i}$, upon receiving $\B{y}_{i+T}$. The overall coding rate of the network, denoted by $\mathsf{R}_{(s,r,d)}$, is given by \cite{KhistiThreeNode}:
 \vspace{-1ex}
\bea
    \mathsf{R}_{(s,r,d)} \triangleq \frac{k}{\max\{n_{1},n_{2}\}}.
\eea
 \vspace{-1ex}

A packet-level code over the three-node $\mathsf{D\&F}$ relay network will be referred to as an $(a_{1},b_{1},a_{2},b_{2},T)$ streaming code if it can simultaneously recover under the overall decoding-delay $T$ from all the permissible erasure patterns of the $(s,r)$ and $(r,d)$ channels. It is entirely possible that the packet-level code employs a coding technique that causes different message symbols lying within the same packet to be decoded with different delays. For example, this is the case with the diagonal embedding and staggered diagonal embedding approaches employed in \cite{NikPVK, KhistiExplicitCode, Small, NikRamVajKum, RamVajNikPVK, GSDE}. This suggests that as shown in Fig. \eqref{fig:three-node}, $(r,d)$ encoding scheme may group together message symbols belonging to different $(s,r)$ packets to form the message symbols that are part of the same packet transmitted from the relay to destination. Such an approach was successfully adopted in \cite{KhistiThreeNode} for the case when the $(s,r)$ and $(r,d)$ channels encountered only arbitrary erasures. This approach was termed as Symbol-Wise Decode-and-Forward ($\mathsf{SW~D\&F}$). With this in mind, we introduce the concept of delay profile \cite{KhistiThreeNode}.
 \begin{definition} A delay profile for a $k$-length message packet is defined as 
 \vspace{-2ex}
\bean
    \textbf{d}=\Big((t_{0}, \tau_{0}),(t_{1}, \tau_{1}),\ldots,(t_{k-1}, \tau_{k-1})\Big),
\eean
\vspace{-3ex}

where $ (t_{\ell}, \tau_{\ell}) \in \mathbb{Z}_{+}^{2}$ denotes the decoding delay of $\B{m}_i(\ell)$ at $(s,r)$ and $(r,d)$ channels, respectively, for any time $i$. The first-hop and second-hop delay profiles are defined to be
$(t_{0}, t_{1}, \ldots, t_{k-1})$ and $(\tau_{0}, \tau_{1}, \ldots, \tau_{k-1})$, respectively. 
\end{definition}
\vspace{-3ex}
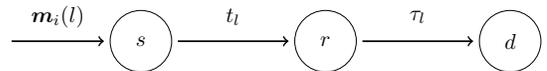
\begin{figure}[H]
\centering
\resizebox{.4 \textwidth}{!}
{
\begin{tikzpicture}
\draw[thick,->] (-2.1,0) -- (-0.6,0);
\draw (-1.35,0.4) node{$\B{m}_{i}(l)$};
\draw (0,0) circle (0.5cm);
\draw (0,0) node{$s$};
\draw[thick,->] (0.6,0) -- (2.4,0);
\draw (1.5,0.4) node{$t_l$};
\draw (3,0) circle (0.5cm);
\draw (3,0) node{$r$};
\draw[thick,->] (3.6,0) -- (5.4,0);
\draw (4.5,0.4) node{$\tau_l$};
\draw (6,0) circle (0.5cm);
\draw (6,0) node{$d$};
\end{tikzpicture}
}
\caption{Illustrating a delay profile for the three-node $\mathsf{D\&F}$ relay network. Note that the decoding delays are independent of $i$.} \label{fig:three-node}
\end{figure}

Thus, the overall decoding delay can be interpreted as requiring that the delay profile satisfies the following constraint 
\bea\label{eq:delprfl}
    t_{\ell} + \tau_{\ell} \leq T,
\eea
where  $\ell \in [0,k-1]$.


\section{An Upper bound on Three-node Decode-and-Forward Relay Network Capacity}
\label{sec:upp-bnd}
In this section, we first define the capacity of the three-node $\mathsf{D\&F}$ relay network and then derive an upper bound on it. 
\begin{definition} The capacity of three-node $\mathsf{D\&F}$ relay network, denoted by $\mathsf{C}_{a_{1},b_{1},a_{2},b_{2},T}$, is the maximum rate achievable by $(a_{1},b_{1},a_{2},b_{2},T)$ streaming codes, i.e.,
\bean
    \mathsf{C}_{a_{1},b_{1},a_{2},b_{2},T} \triangleq \sup \left\{
    \begin{minipage}{0.095\textwidth}
    $\frac{k}{\max\{n_{1},n_{2}\}} $
    \end{minipage} 
    \middle|
    ~\begin{minipage}{0.125\textwidth}
    \small{
    There exists an \\ 
    $(a_{1},b_{1},a_{2},b_{2},T)$\\ 
    streaming code \\
    with parameters\\
    $(k,n_{1},n_{2})_{\mathbb{F}_{q}} $}
    \end{minipage} \right\}.
\eean
\end{definition}

\begin{theorem}
For any $(a_{1},b_{1},a_{2},b_{2},T)$,
\bea\label{eq:cap}
  \mathsf{C}_{a_{1},b_{1},a_{2},b_{2},T} 
  \leq \min\{\mathsf{C}_{a_{1},b_{1},T-b_{2}},\mathsf{C}_{a_{2},b_{2},T-b_{1}}\},
\eea
where $\mathsf{C}_{a_{i},b_{i},T_{i}}$ are the point-to-point channel capacities. 
\end{theorem}
\begin{proof}

The bound can be derived as follows. If the source transmits at a rate in excess of $\mathsf{C}_{a_{1},b_{1},T-b_{2}}$, then there is at least one permissible erasure pattern which will make it impossible to decode a packet transmitted at time $i$ by time $i+T-b_2$ at the relay. If the $(r,d)$ channel then experiences a burst of duration $b_2$, then this will make it impossible for the destination to recover packet $i$ by time $(i+T)$. Therefore, $  \mathsf{C}_{a_{1},b_{1},a_{2},b_{2},T} \leq \mathsf{C}_{a_{1},b_{1},T-b_{2}}$. The upper bound $  \mathsf{C}_{a_{1},b_{1},a_{2},b_{2},T} \leq \mathsf{C}_{a_{2},b_{2},T-b_{1}}$ follows from noting that an initial burst of duration $b_1$ on the $(s,r)$ channel, is equivalent to reducing the time available for the relay to convey information to the destination by an amount $b_1$. 
\end{proof}

\section{Staggered Diagonal Embedding}
\label{sec:SDE}
We use the staggered diagonal embedding ($\mathsf{SDE}$) approach introduced in \cite{NikRamVajKum} to construct a rate-optimal code over the three-node $\mathsf{D\&F}$ relay network. As shown in Section \ref{sec:constr}, the $\mathsf{SDE}$ approach enables matching of channel-level delay profiles to satisfy the overall decoding delay constraint. 
\begin{figure}[H]
\centering

\resizebox{0.4 \textwidth}{!}
{
\begin{tikzpicture}
\fill (-0.75,3) circle[radius=1pt];
\fill (-0.5,3) circle[radius=1pt];
\fill (-0.25,3) circle[radius=1pt];
\draw[step=1cm,black,very thin] (0,0) grid (9,6);
\draw (0.5,5.5) node{$c_0$};
\draw [ultra thick, draw=black, fill=orange, opacity=0.35] (0,5) rectangle (1,6); 
\draw (1.5,4.5) node{$c_1$};
\draw [ultra thick, draw=black, fill=orange, opacity=0.35] (1,4) rectangle (2,5); 
\draw (2.5,3.5) node{$c_2$};
\draw [ultra thick, draw=black, fill=orange, opacity=0.35] (2,3) rectangle (3,4); 
\draw (6.5,2.5) node{$c_3$};
\draw [ultra thick, draw=black, fill=orange, opacity=0.35] (6,2) rectangle (7,3); 
\draw (7.5,1.5) node{$c_4$};
\draw [ultra thick, draw=black, fill=orange, opacity=0.35] (7,1) rectangle (8,2); 
\draw (8.5,0.5) node{$c_5$};
\draw [ultra thick, draw=black, fill=orange, opacity=0.35] (8,0) rectangle (9,1); 
\draw [ultra thick, draw=black, opacity=0.5, pattern=north west lines] (1,-0.5) rectangle (7,6.5); 
\draw (4,6.8) node{burst of length $6$};
\draw (4.5,-1) node{each column represents a packet };
\fill (9.25,3) circle[radius=1pt];
\fill (9.5,3) circle[radius=1pt];
\fill (9.75,3) circle[radius=1pt];
\end{tikzpicture}
}
\caption{Staggered diagonal embedding of a codeword in $[6,3] \mathsf{MDS}$-base code within the packet stream. The contents of the codeword can be successfully recovered as a burst of length $6$ erases at most $3$ symbols of the diagonally-embedded codeword.} 
\label{fig:SS-code}
\end{figure}
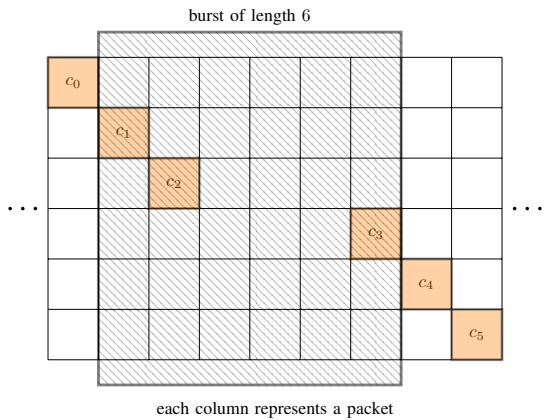




In \cite{NikRamVajKum}, for an $(a,b,T)$ $\mathsf{DCSW}$ channel, $\mathsf{SDE}$ approach was used to construct a rate-optimal streaming code for the parameter regime: $b~|~T-a+1$, with field size \textit{linear} in $T$. We briefly describe the $\mathsf{SDE}$ approach below.

Let $\mathcal{C}$ be an $[n,k]$ linear code in systematic form, with the first $k$ code symbols corresponding to the message symbols. Let $N \ge n$ be an integer and let $S$ be a subset of $[0:N-1]$ of size $n$: 
\bean
S=\{i_0, i_1, \dots , i_{n-1}\} \subseteq[0:N-1], 
\eean
where $0\triangleq i_0 < i_1 < \dots < i_{n-1}\triangleq N-1$. Let $\B{x}_{i}(\ell)$ denote the $(\ell+1)^{th}$ symbol of the coded packet $\B{x}_{i}$ for each $\ell \in {0, 1,\ldots,n-1} $.  Then for every time $i$, the collection of symbols 
\bean
\Big(\B{x}_{i+i_0}(0), \B{x}_{i+i_1}(1),\cdots,\B{x}_{i+i_{n-1}}(n-1)\Big)
\eean
forms a codeword in code $\mathcal{C}$. The set $S$ is called as the placement set as it determines the placement in time of the code symbols of $\mathcal{C}$ within the packet stream, and $N$ is referred to as the dispersion-span parameter. The packet-level code thus constructed through $\mathsf{SDE}$, has rate $\frac{k}{n}$. The code ${\cal C}$ is referred to as the base code.

\subsection{Construction of an $(a,b,T)$ streaming code using $\mathsf{SDE}$ of an $\mathsf{MDS}$-Base-Code}
In \cite{NikRamVajKum}, it was shown that the $\mathsf{SDE}$ approach yields $(a,b,T)$ streaming codes meeting the dispersion-span constraint $N\leq T+1$ for all parameters $\{a,b,T\}$. Let $T+1=mb+\delta_1$, with $0 \le \delta_1 < b$ and set $\delta_2 \triangleq \min\{\delta_1, a\}$. Let the base code $\mathcal{C}$ be an $[n=k+a ,k]$ $\mathsf{MDS}$ code, where $k \in \mathbb{Z}_{+}$. The authors of \cite{NikRamVajKum} selected the value of the parameter $k$ to be the maximum possible, which is $(m-1)a + \delta_2$ but for ensuring optimality over the three-node $\mathsf{D\&F}$ relay network we require a wider range for $k$, thus we assume that $k \leq (m-1)a + \delta_2$. Additionally let the dispersion-span parameter and dispersion set be given by: 
\bea
\label{eq:disp-span}
N =  n + (b-a).\floor{\frac{n-1}{a}},\\
\label{eq:disp-set}
S =\displaystyle\bigcup_{i=0}^{n-1} \left\{i+(b-a).\floor{\frac{i}{a}}\right\},
\eea
respectively. Since $k \leq (m-1)a + \delta_2$, it follows that
\begin{align*}
N &\leq ma +\delta_2 + (b-a).\floor{\frac{ma+\delta_2-1}{a}}\\
&= mb + \delta_2 \\
&\leq mb+\delta_1 =  T+1.
\end{align*}
The delay profile simply follows from the dispersion-span parameter and the displacement set; the decoding delay of the $j^{\text{th}}$ message symbol is
\bea \label{ds}
t_j = N-1-j-(b-a).\floor{\frac{j}{a}} \leq T,
\eea
where $j \in [0,k-1]$. Also, as shown in Fig. \eqref{fig:SS-code}, it can be seen that the erasure of any successive $b$ packets corresponds to the erasure of exactly $a$ code symbols belonging to the base code $\mathcal{C}$. Thus, due to the $\mathsf{MDS}$ nature of ${\cal C}$, the packet-level code can tolerate a burst of length $b$ and any $a$ random erasures. Therefore, the packet-level code constructed using the $\mathsf{SDE}$ approach is an $(a,b,T)$ streaming code. As the base code is an MDS code, the code construction requires \textit{linear} field size.

\begin{figure*}
\centering
\begin{subfigure}{0.9 \textwidth}
\resizebox{1 \textwidth}{!}
{
\begin{tikzpicture}
\draw [thick, draw=black] (-1,0) rectangle (2,0.5);
\draw [thick, draw=black] (-1,-0.5) rectangle (2,0);
\draw [thick, draw=black] (-1,-1) rectangle (2,-0.5);
\draw [thick, draw=black] (-1,-1.5) rectangle (2,-1);

\draw [thick, draw=black] (2,0) rectangle (4.5,0.5);
\draw [thick, draw=black] (2,-0.5) rectangle (4.5,0);
\draw [thick, draw=black] (2,-1) rectangle (4.5,-0.5);
\draw [thick, draw=black] (2,-1.5) rectangle (4.5,-1);

\draw [thick, draw=black] (4.5,0) rectangle (7,0.5);
\draw [thick, draw=black] (4.5,-0.5) rectangle (7,0);
\draw [thick, draw=black] (4.5,-1) rectangle (7,-0.5);
\draw [thick, draw=black] (4.5,-1.5) rectangle (7,-1);

\draw [thick, draw=black] (7,0) rectangle (9.5,0.5);
\draw [thick, draw=black] (7,-0.5) rectangle (9.5,0);
\draw [thick, draw=black] (7,-1) rectangle (9.5,-0.5);
\draw [thick, draw=black] (7,-1.5) rectangle (9.5,-1);

\draw [thick, draw=black] (9.5,0) rectangle (12,0.5);
\draw [thick, draw=black] (9.5,-0.5) rectangle (12,0);
\draw [thick, draw=black] (9.5,-1) rectangle (12,-0.5);
\draw [thick, draw=black] (9.5,-1.5) rectangle (12,-1);

\draw [thick, draw=black] (12,0) rectangle (14.5,0.5);
\draw [thick, draw=black] (12,-0.5) rectangle (14.5,0);
\draw [thick, draw=black] (12,-1) rectangle (14.5,-0.5);
\draw [thick, draw=black] (12,-1.5) rectangle (14.5,-1);

\draw [thick, draw=black] (14.5,0) rectangle (17,0.5);
\draw [thick, draw=black] (14.5,-0.5) rectangle (17,0);
\draw [thick, draw=black] (14.5,-1) rectangle (17,-0.5);
\draw [thick, draw=black] (14.5,-1.5) rectangle (17,-1);

\draw [thick, draw=black] (17,0) rectangle (19.5,0.5);
\draw [thick, draw=black] (17,-0.5) rectangle (19.5,0);
\draw [thick, draw=black] (17,-1) rectangle (19.5,-0.5);
\draw [thick, draw=black] (17,-1.5) rectangle (19.5,-1);

\draw (0.5,0.25) node{time $i$};
\draw (0.5,-0.25) node{$\B{m}_i[0]$};
\draw (0.5,-0.75) node{$\B{m}_{i}[1]$};
\draw (0.5,-1.25) node{$\B{m}_{i-4}[0] + \B{m}_{i-2}[1]$};
\draw (3.25,0.25) node{$10$};
\draw (3.25,-0.25) node{$\B{m}_{10}[0]$};
\draw (3.25,-0.75) node{$\B{m}_{10}[1]$};
\draw (3.25,-1.25) node{$\B{m}_{6}[0] + \B{m}_{8}[1]$};
\draw (5.75,0.25) node{$11$};
\draw (5.75,-0.25) node{$\B{m}_{11}[0]$};
\draw (5.75,-0.75) node{$\B{m}_{11}[1]$};
\draw (5.75,-1.25) node{$\B{m}_{7}[0] + \B{m}_{9}[1]$};
\draw (8.25,0.25) node{$12$};
\draw (8.35,-0.25) node{$\B{m}_{12}[0]$};
\draw (8.25,-0.75) node{$\B{m}_{12}[1]$};
\draw (8.25,-1.25) node{$\B{m}_{8}[0] + \B{m}_{10}[1]$};
\draw (10.75,0.25) node{$13$};
\draw (10.75,-0.25) node{$\B{m}_{13}[0]$};
\draw (10.75,-0.75) node{$\B{m}_{13}[1]$};
\draw (10.75,-1.25) node{$\B{m}_{9}[0] + \B{m}_{11}[1]$};
\draw (13.25,0.25) node{$14$};
\draw (13.25,-0.25) node{$\B{m}_{14}[0]$};
\draw (13.25,-0.75) node{$\B{m}_{14}[1]$};
\draw (13.25,-1.25) node{$\B{m}_{10}[0] + \B{m}_{12}[1]$};
\draw (15.75,0.25) node{$15$};
\draw (15.75,-0.25) node{$\B{m}_{15}[0]$};
\draw (15.75,-0.75) node{$\B{m}_{15}[1]$};
\draw (15.75,-1.25) node{$\B{m}_{11}[0] + \B{m}_{13}[1]$};
\draw (18.25,0.25) node{$16$};
\draw (18.25,-0.25) node{$\B{m}_{16}[0]$};
\draw (18.25,-0.75) node{$\B{m}_{16}[1]$};
\draw (18.25,-1.25) node{$\B{m}_{12}[0] + \B{m}_{14}[1]$};

\draw [ultra thick, draw=black, fill=red, opacity=0.2] (2,-0.5) rectangle (4.5,0); 
\draw [ultra thick, draw=black, fill=red, opacity=0.2] (7,-1) rectangle (9.5,-0.5);
\draw [ultra thick, draw=black, fill=red, opacity=0.2] (12,-1.5) rectangle (14.5,-1);

\draw [ultra thick, draw=black, fill=blue, opacity=0.2] (4.5,-0.5) rectangle (7,0); 
\draw [ultra thick, draw=black, fill=blue, opacity=0.2] (9.5,-1) rectangle (12,-0.5);
\draw [ultra thick, draw=black, fill=blue, opacity=0.2] (14.5,-1.5) rectangle (17,-1);

\draw [ultra thick, draw=black, fill=green, opacity=0.2] (7,-0.5) rectangle (9.5,0); 
\draw [ultra thick, draw=black, fill=green, opacity=0.2] (12,-1) rectangle (14.5,-0.5);
\draw [ultra thick, draw=black, fill=green, opacity=0.2] (17,-1.5) rectangle (19.5,-1);
\end{tikzpicture}
}
\caption{Transmitted symbols at node $s$ from time $10$ to $16$} 
\label{fig:code-sr}
\end{subfigure}

\vspace{1em}

\begin{subfigure}{1\textwidth}
\resizebox{1 \textwidth}{!}
{
\begin{tikzpicture}
\draw [thick, draw=black] (-6,0) rectangle (2,0.5);
\draw [thick, draw=black] (-6,-0.5) rectangle (2,0);
\draw [thick, draw=black] (-6,-1) rectangle (2,-0.5);
\draw [thick, draw=black] (-6,-1.5) rectangle (2,-1);

\draw [thick, draw=black] (-3,0) rectangle (-0.5,0.5);
\draw [thick, draw=black] (-3,-0.5) rectangle (-0.5,0);
\draw [thick, draw=black] (-3,-1) rectangle (-0.5,-0.5);
\draw [thick, draw=black] (-3,-1.5) rectangle (-0.5,-1);

\draw [thick, draw=black] (-0.5,0) rectangle (2,0.5);
\draw [thick, draw=black] (-0.5,-0.5) rectangle (2,0);
\draw [thick, draw=black] (-0.5,-1) rectangle (2,-0.5);
\draw [thick, draw=black] (-0.5,-1.5) rectangle (2,-1);

\draw [thick, draw=black] (2,0) rectangle (4.5,0.5);
\draw [thick, draw=black] (2,-0.5) rectangle (4.5,0);
\draw [thick, draw=black] (2,-1) rectangle (4.5,-0.5);
\draw [thick, draw=black] (2,-1.5) rectangle (4.5,-1);

\draw [thick, draw=black] (4.5,0) rectangle (7,0.5);
\draw [thick, draw=black] (4.5,-0.5) rectangle (7,0);
\draw [thick, draw=black] (4.5,-1) rectangle (7,-0.5);
\draw [thick, draw=black] (4.5,-1.5) rectangle (7,-1);

\draw [thick, draw=black] (7,0) rectangle (9.5,0.5);
\draw [thick, draw=black] (7,-0.5) rectangle (9.5,0);
\draw [thick, draw=black] (7,-1) rectangle (9.5,-0.5);
\draw [thick, draw=black] (7,-1.5) rectangle (9.5,-1);

\draw [thick, draw=black] (9.5,0) rectangle (12,0.5);
\draw [thick, draw=black] (9.5,-0.5) rectangle (12,0);
\draw [thick, draw=black] (9.5,-1) rectangle (12,-0.5);
\draw [thick, draw=black] (9.5,-1.5) rectangle (12,-1);

\draw [thick, draw=black] (12,0) rectangle (14.5,0.5);
\draw [thick, draw=black] (12,-0.5) rectangle (14.5,0);
\draw [thick, draw=black] (12,-1) rectangle (14.5,-0.5);
\draw [thick, draw=black] (12,-1.5) rectangle (14.5,-1);

\draw [thick, draw=black] (14.5,0) rectangle (17,0.5);
\draw [thick, draw=black] (14.5,-0.5) rectangle (17,0);
\draw [thick, draw=black] (14.5,-1) rectangle (17,-0.5);
\draw [thick, draw=black] (14.5,-1.5) rectangle (17,-1);

\draw [thick, draw=black] (17,0) rectangle (19.5,0.5);
\draw [thick, draw=black] (17,-0.5) rectangle (19.5,0);
\draw [thick, draw=black] (17,-1) rectangle (19.5,-0.5);
\draw [thick, draw=black] (17,-1.5) rectangle (19.5,-1);

\draw (-4.5,0.25) node{time $i$};
\draw (-4.5,-0.25) node{$\B{m}_{i-2}[1]$};
\draw (-4.5,-0.75) node{$\B{m}_{i-4}[0]$};
\draw (-4.5,-1.25) node{$\B{m}_{i-8}[1] + \B{m}_{i-7}[0]$};
\draw (-1.75,0.25) node{$12$};
\draw (-1.75,-0.25) node{$\B{m}_{10}[1]$};
\draw (-1.75,-0.75) node{$\B{m}_{8}[0]$};
\draw (-1.75,-1.25) node{$\B{m}_{4}[1] + \B{m}_{5}[0]$};
\draw (0.75,0.25) node{$13$};
\draw (0.75,-0.25) node{$\B{m}_{11}[1]$};
\draw (0.75,-0.75) node{$\B{m}_{9}[0]$};
\draw (0.75,-1.25) node{$\B{m}_{5}[1] + \B{m}_{6}[0]$};
\draw (3.25,0.25) node{$14$};
\draw (3.25,-0.25) node{$\B{m}_{12}[1]$};
\draw (3.25,-0.75) node{$\B{m}_{10}[0]$};
\draw (3.25,-1.25) node{$\B{m}_{6}[1] + \B{m}_{7}[0]$};
\draw (5.75,0.25) node{$15$};
\draw (5.75,-0.25) node{$\B{m}_{13}[1]$};
\draw (5.75,-0.75) node{$\B{m}_{11}[0]$};
\draw (5.75,-1.25) node{$\B{m}_{7}[1] + \B{m}_{8}[0]$};
\draw (8.25,0.25) node{$16$};
\draw (8.35,-0.25) node{$\B{m}_{14}[1]$};
\draw (8.25,-0.75) node{$\B{m}_{12}[0]$};
\draw (8.25,-1.25) node{$\B{m}_{8}[1] + \B{m}_{9}[0]$};
\draw (10.75,0.25) node{$17$};
\draw (10.75,-0.25) node{$\B{m}_{15}[1]$};
\draw (10.75,-0.75) node{$\B{m}_{13}[0]$};
\draw (10.75,-1.25) node{$\B{m}_{9}[1] + \B{m}_{10}[0]$};
\draw (13.25,0.25) node{$18$};
\draw (13.25,-0.25) node{$\B{m}_{16}[1]$};
\draw (13.25,-0.75) node{$\B{m}_{14}[0]$};
\draw (13.25,-1.25) node{$\B{m}_{10}[1] + \B{m}_{11}[0]$};
\draw (15.75,0.25) node{$19$};
\draw (15.75,-0.25) node{$\B{m}_{17}[1]$};
\draw (15.75,-0.75) node{$\B{m}_{15}[0]$};
\draw (15.75,-1.25) node{$\B{m}_{11}[1] + \B{m}_{12}[0]$};
\draw (18.25,0.25) node{$20$};
\draw (18.25,-0.25) node{$\B{m}_{18}[1]$};
\draw (18.25,-0.75) node{$\B{m}_{16}[0]$};
\draw (18.25,-1.25) node{$\B{m}_{12}[1] + \B{m}_{13}[0]$};

\draw [ultra thick, draw=black, fill=yellow, opacity=0.2] (-3,-0.5) rectangle (-0.5,0); 
\draw [ultra thick, draw=black, fill=yellow, opacity=0.2] (4.5,-1) rectangle (7,-0.5);
\draw [ultra thick, draw=black, fill=yellow, opacity=0.2] (12,-1.5) rectangle (14.5,-1);

\draw [ultra thick, draw=black, fill=magenta, opacity=0.2] (-0.5,-0.5) rectangle (2,0); 
\draw [ultra thick, draw=black, fill=magenta, opacity=0.2] (7,-1) rectangle (9.5,-0.5);
\draw [ultra thick, draw=black, fill=magenta, opacity=0.2] (14.5,-1.5) rectangle (17,-1);

\draw [ultra thick, draw=black, fill=orange, opacity=0.2] (2,-0.5) rectangle (4.5,0); 
\draw [ultra thick, draw=black, fill=orange, opacity=0.2] (9.5,-1) rectangle (12,-0.5);
\draw [ultra thick, draw=black, fill=orange, opacity=0.2] (17,-1.5) rectangle (19.5,-1);
\end{tikzpicture}
}
\caption{Transmitted symbols at node $r$ from time $12$ to $20$} 
\label{fig:code-rd}
\end{subfigure}

\vspace{1em}

\begin{subfigure}{0.9\textwidth}
\resizebox{1 \textwidth}{!}
{
\begin{tikzpicture}
\draw [thick, draw=black] (-6,0) rectangle (2,0.5);
\draw [thick, draw=black] (-6,-0.5) rectangle (2,0);
\draw [thick, draw=black] (-6,-1) rectangle (2,-0.5);

\draw [thick, draw=black] (-3,0) rectangle (-0.5,0.5);
\draw [thick, draw=black] (-3,-0.5) rectangle (-0.5,0);
\draw [thick, draw=black] (-3,-1) rectangle (-0.5,-0.5);

\draw [thick, draw=black] (-0.5,0) rectangle (2,0.5);
\draw [thick, draw=black] (-0.5,-0.5) rectangle (2,0);
\draw [thick, draw=black] (-0.5,-1) rectangle (2,-0.5);

\draw [thick, draw=black] (2,0) rectangle (4.5,0.5);
\draw [thick, draw=black] (2,-0.5) rectangle (4.5,0);
\draw [thick, draw=black] (2,-1) rectangle (4.5,-0.5);

\draw [thick, draw=black] (4.5,0) rectangle (7,0.5);
\draw [thick, draw=black] (4.5,-0.5) rectangle (7,0);
\draw [thick, draw=black] (4.5,-1) rectangle (7,-0.5);

\draw [thick, draw=black] (7,0) rectangle (9.5,0.5);
\draw [thick, draw=black] (7,-0.5) rectangle (9.5,0);
\draw [thick, draw=black] (7,-1) rectangle (9.5,-0.5);

\draw [thick, draw=black] (9.5,0) rectangle (12,0.5);
\draw [thick, draw=black] (9.5,-0.5) rectangle (12,0);
\draw [thick, draw=black] (9.5,-1) rectangle (12,-0.5);

\draw [thick, draw=black] (12,0) rectangle (14.5,0.5);
\draw [thick, draw=black] (12,-0.5) rectangle (14.5,0);
\draw [thick, draw=black] (12,-1) rectangle (14.5,-0.5);

\draw (-4.5,0.25) node{time $i$};
\draw (-4.5,-0.25) node{$\B{m}_{i-8}[1]$};
\draw (-4.5,-0.75) node{$\B{m}_{i-7}[0]$};
\draw (-1.75,0.25) node{$15$};
\draw (-1.75,-0.25) node{$\B{m}_{7}[1]$};
\draw (-1.75,-0.75) node{$\B{m}_{8}[0]$};
\draw (0.75,0.25) node{$16$};
\draw (0.75,-0.25) node{$\B{m}_{8}[1]$};
\draw (0.75,-0.75) node{$\B{m}_{9}[0]$};
\draw (3.25,0.25) node{$17$};
\draw (3.25,-0.25) node{$\B{m}_{9}[1]$};
\draw (3.25,-0.75) node{$\B{m}_{10}[0]$};
\draw (5.75,0.25) node{$18$};
\draw (5.75,-0.25) node{$\B{m}_{10}[1]$};
\draw (5.75,-0.75) node{$\B{m}_{11}[0]$};
\draw (8.25,0.25) node{$19$};
\draw (8.35,-0.25) node{$\B{m}_{11}[1]$};
\draw (8.25,-0.75) node{$\B{m}_{12}[0]$};
\draw (10.75,0.25) node{$20$};
\draw (10.75,-0.25) node{$\B{m}_{12}[1]$};
\draw (10.75,-0.75) node{$\B{m}_{13}[0]$};
\draw (13.25,0.25) node{$21$};
\draw (13.25,-0.25) node{$\B{m}_{13}[1]$};
\draw (13.25,-0.75) node{$\B{m}_{14}[0]$};
\end{tikzpicture}
}
\caption{Recovered symbols at node $d$ from time $15$ to $21$} 
\label{fig:code-rd}
\end{subfigure}
\vspace{1em}
\caption{Rate-optimal $(a_1=1, b_1=2, a_2 =1, b_2=3, T=8)$ streaming code}
\label{fig:constr}
\end{figure*}

\section{Rate-optimal Code Construction for the Three-Node $\mathsf{D\&F}$ Relay Network}
\label{sec:constr}

We first present an example, shown in Fig. \eqref{fig:constr}, to demonstrate how $\mathsf{SDE}$ approach and $\mathsf{SW~D\&F}$ strategy are employed to construct a rate-optimal $(a_1=1, b_1=2, a_2 =1, b_2=3, T=8)$ streaming code. From \eqref{eq:cap}, the capacity is upper bounded by 
\begin{align*}
\mathsf{C}_{a_1,b_1,a_2,b_2,} &\leq \min\{\mathsf{C}_{a_{1},b_{1},T-b_{2}},\mathsf{C}_{a_{2},b_{2},T-b_{1}}\}\\
&=\min\left\{\frac{8-3-1+1}{8-3-1+1+2},\frac{8-2-1+1}{8-2-1+1+3}\right\} \\
&= \min\left\{\frac{5}{7},\frac{6}{9}\right\} = \frac{2}{3}.
\end{align*}
For the $(s,r)$ channel, we use an $[3,2]$ $\mathsf{MDS}$ code as the base code to construct $(a_1=1,b_1=2,T_1=5)$ streaming code. Similarly, for the $(r,d)$ channel, we use an $[3,2]$ $\mathsf{MDS}$ code as the base code to construct $(a_2=1,b_2=3,T_2=6)$ streaming code. While re-encoding the decoded message symbols at the relay node, the order of the packet symbols within a packet is flipped. Therefore, from \eqref{ds}, the delay profile of the packet is 
\bean
\big((t_0 = 4,\tau_0 = 3),(t_1= 2,\tau_1 = 6)\big).
\eean
From the delay profile, it can be seen that the overall decoding delay for each packet symbol is $\leq 8$ packets. Therefore, the code construction is a rate-optimal $(a_1=1, b_1=2, a_2 =1, b_2=3, T=8)$ streaming code.

To construct a rate-optimal $(a_{1},a_{2},b_{1},b_{2},T)$ streaming code over the three-node $\mathsf{D\&F}$ relay network, we require the parameters $k,n_1,n_2$ to satisfy:
\begin{align*}
&\mathsf{R}_{(s,r,d)} = \frac{k}{\max\{n_{1},n_{2}\}} = \min\{\mathsf{C}_{a_{1},b_{1},T-b_{2}},\mathsf{C}_{a_{2},b_{2},T-b_{1}}\}.
\end{align*}
It can be shown that this leads to:
\begin{equation}\label{rok}
\resizebox{.43 \textwidth}{!}
{$ k = \max\left\{a'_{1},a'_{2}\right\}.\min\left\{\frac{T-b_{2}-a_{1}+1}{b_{1}},\frac{T-b_{1}-a_{2}+1}{b_{2}}\right\}$,}
\end{equation} 
where $a'_{1} = n_1 - k$ and  $a'_{2} = n_2 -k$. We now describe the code construction. 
\begin{constr}\label{cons}~\\
Let $a\triangleq\max\{a_1,a_2\},b'_u \triangleq\max\{b_u,a\},~\alpha \triangleq T+1-b'_{1}-b'_{2}-a$, for $u\in[1,2]$. 
We select $T_1=T-b_2'$ and $T_2 = T-b_1'$ as the channel-level decoding delay constraints for the $(s,r)$ and $(r,d)$ channels, respectively. Thus, the $(s,r)$ and $(r,d)$ channels are modeled as $(a_1,b_1,T-b_2')$ DCSW and $(a_2,b_2,T-b_1')$ DCSW channels. 
We then construct $(a,b_1',T-b_2')$ and $(a,b_2',T-b_1')$ streaming codes for the $(s,r)$ and $(r,d)$ channels, respectively, using $\mathsf{SDE}$ of $[\floor{k}+a , \floor{k}]$ $\mathsf{MDS}$ base codes for both the channels, where $k$ is
\begin{align*}
&k = \max\left\{a,a\right\}.\min\left\{\frac{T-b'_{2}-a+1}{b'_{1}},\frac{T-b'_{1}-a+1}{b'_{2}}\right\}\\
&= a.\min\left\{\frac{\alpha}{b'_{1}}+1,\frac{\alpha}{b'_{2}}+1\right\}= a.\frac{T-\min\{b'_{1},b'_{2}\}-a+1}{\max\{b'_{1},b'_{2}\}}. 
\end{align*}
Let $u \in[1:2]$, from \eqref{eq:disp-span}, the dispersion-span parameters are
\begin{align*}
    N_u = n_u + (b_u'-a).\floor{\frac{n_u-1}{a}}.
\end{align*}
While re-encoding the message symbols at the relay node, the order of the message symbols within a packet is flipped following the approach adopted in \cite{KhistiThreeNode}. From \eqref{ds}, it follows that the decoding delays of the $j^{\text{th}}$ message symbol across the $(s,r)$ and $(r,d)$ channels are given by
\begin{align*}
t_j &= N_1-1-j-(b_1'-a).\floor{\frac{j}{a}},\\
\tau_j &= N_2-1-(k-1-j)-(b_2'-a).\floor{\frac{k-1-j}{a}},
\end{align*}
respectively, where $j \in [0:k-1]$. 
\end{constr}

We next present a lemma involving numerics.

\begin{lemma}\label{pr:2} Let $x,y,z \in \mathbb{Z}, z>0$. Let $x=q_{1}z+r_{1}$ and $y=q_{2}z+r_{2}$, where $q_{1}=\floor{\frac{x}{z}}, r_{1}\equiv x \bmod z, q_{2}=\floor{\frac{y}{z}}, r_{2}\equiv y \bmod z.$
\begin{align*}
    \floor{\frac{x}{z}-\frac{y}{z}} 
    = \begin{cases} 
      q_{1}-q_{2} &  r_{1}\geq r_{2} \\
      q_{1}-q_{2}-1 &  r_{1} < r_{2}.
  \end{cases}
\end{align*}
\end{lemma}
\begin{proof}
\begin{align*}
    \floor{\frac{x}{z}-\frac{y}{z}} &=\floor{\frac{q_{1}z+r_{1}}{z}-\frac{q_{2}z+r_{2}}{z}}\\
    &=\floor{q_{1}-q_{2}+\frac{r_{1}}{z}-\frac{r_{2}}{z}}\\
    &=q_{1}-q_{2}+\floor{\frac{r_{1}}{z}-\frac{r_{2}}{z}}
    \intertext{from which the result follows}
    &= \begin{cases} 
      q_{1}-q_{2} &  r_{1}\geq r_{2} \\
      q_{1}-q_{2}-1 &  r_{1} < r_{2}.
  \end{cases}
\end{align*}
\end{proof}

\begin{theorem}
If $T,a_1,a_2,b_1,b_2$ satisfy:
\begin{align*}
    \max\{b'_{1},b'_{2}\} ~ &| ~ \alpha,~ a_1=a_2 ~\text{OR} ~b_1=b_2
\end{align*}
then construction \ref{cons} is a rate-optimal $(a_{1},a_{2},b_{1},b_{2},T)$ streaming code.
\end{theorem}
\begin{proof}
Since $\max\{b'_{1},b'_{2}\} | (T-\min\{b'_{1},b'_{2}\}+1-a)$, 
\bean
\floor{k} = k.
\eean
Therefore, when $a_1=a_2$ or $b_1=b_2$, Construction \ref{cons} in this parameter regime is rate-optimal. Since, in this parameter regime $a\mid k$, it follows that
\begin{align*}
    N_u &= n_u + (b_u'-a)\left(\frac{n_u}{a}-1\right)\\
    &= k+a + (b_u'-a)\frac{k}{a} = a + b_u'\frac{k}{a}.
\end{align*}
It can be verified that the dispersion-span constraint is met for the both the channels. From \eqref{eq:delprfl}, we have that for $j\in[0,k-1], t_j + \tau_j \leq T$ or equivalently,
\begin{align*}
\resizebox{0.48\textwidth}{!}
{$N_1 + N_2 -k-1-(b_1'-a).\floor{\frac{j}{a}}-(b_2'-a).\floor{\frac{k-1-j}{a}} \leq T$.}
\end{align*}
Consider the  L.H.S. $\triangleq f(j)$  of the above inequality,
\begin{align*}
\resizebox{.49 \textwidth}{!}
{
$f(j) = N_1+N_2-k-1-(b_1'-a)\floor{\frac{j}{a}}-(b_2'-a)\floor{\frac{k-1-j}{a}}.$
}
\end{align*}
Since $a \mid k$, we have that $f(j)$ is equal to 
\begin{align*}
\resizebox{.49 \textwidth}{!}
{
$N_1 + N_2 -k-1 -(b_1'-a)\floor{\frac{j}{a}} - (b_2'-a)\left(\frac{k}{a}+\floor{\frac{-1-j}{a}}\right)$.
}
\end{align*}
From Lemma \ref{pr:2}, it follows that $\floor{\frac{-1}{a}-\frac{j}{a}} = \floor{\frac{-1}{a}} -\floor{\frac{j}{a}} = -1 -\floor{\frac{j}{a}}$. Therefore, 
\begin{align*}
f(j)&= N_1 + N_2 -1-a-b_2'\frac{k}{a} +b_2' - (b_1'-b_2')\floor{\frac{j}{a}}\\
&= N_1 -1 +b_2' - (b_1'-b_2')\floor{\frac{j}{a}}.
\end{align*}
For the inequality to hold for all $j \in [0:k-1]$, the following constriant must be satisfied:
\begin{align*}
\underset{j \in[0:k-1]}{\max}~f(j) ~\leq~ T.
\end{align*}
\textbf{Case I: $b_1' \geq b_2'$}~: If $b_1' - b_2' > 0$ then $f(j)$ is maximum at $j=0$, else $b_1' = b_2'$ then $f(j) = N_1 -1 +b_2'$. Therefore, 
\begin{align*}
\underset{j \in[0:k-1]}{\max}~f(j) &=  N_1 -1 +b_2'.
\intertext{It then follows from the dispersion-span constraint that}
\underset{j \in[0:k-1]}{\max}~f(j)&\leq T-b_2'+1 -1+b_2' = T.
\end{align*}
\textbf{Case II: $b_1' < b_2'$}~: As $b_1' - b_2' < 0$, $f(j)$ is maximum when $j=k-1$. Therefore,
\begin{align*}
\underset{j \in[0:k-1]}{\max}~f(j) &= N_1 -1 +b_2' - (b_1'-b_2')\floor{\frac{k-1}{a}}.
\intertext{Since $a \mid k$,}
\underset{j \in[0:k-1]}{\max}~f(j)&= N_1 -1 +b_2' - (b_1'-b_2')\left(\frac{k}{a}-1\right)\\
&= N_1 -1 +b_1' + (b_2'-b_1')\left(\frac{k}{a}\right)\\
&= a -1 +b_1' + b_2'\left(\frac{k}{a}\right)\\
&= N_2 - 1 + b_1'.
\intertext{It then follows from the dispersion-span constraint that}
\underset{j \in[0:k-1]}{\max}~f(j)&\leq T-b_1'+1 -1+b_1' = T.
\end{align*}
Therefore, Construction \ref{cons} satisfies the decoding delay constraint. As we have used an $(a,b_1',T-b_2')$ streaming code for the $(s,r)$ channel, the code can recover from any permissible erasures of $(a_1,b_1,T-b_2')$ $\mathsf{DCSW}$ channel. Similarly, for the $(r,d)$ channel, as we have used an $(a,b_2',T-b_1')$ streaming code, the code can recover from any permissible erasures of $(a_2,b_2,T-b_1')$ $\mathsf{DCSW}$ channel. Therefore, Construction \ref{cons} is a rate-optimal $(a_1,a_2,b_1,b_2,T)$ streaming code.
\end{proof}

\section{Conclusions}
\label{sec:conclusion}
An upper bound on the capacity of three-node decode-and-forward relay network was derived, and was shown to be tight for the parameter regime: $\max\{b_1,b_2\}~|~(T-b_1-b_2-\max\{a_1,a_2\}+1), a_1=a_2 ~\text{OR} ~b_1=b_2$. A rate-optimal code construction having a field size \textit{linear} in decoding delay, $T$, was designed by using staggered diagonal embedding approach and symbol-wise $\mathsf{D\&F}$ strategy.

\end{document}